\documentclass[twoside]{article}

\usepackage[accepted]{aistats2020}

\usepackage{algorithm}
\usepackage[noend]{algpseudocode}
\usepackage{amsfonts}
\usepackage{amsmath}
\usepackage{amssymb}
\usepackage{amsthm}
\usepackage{appendix}
\usepackage{booktabs}
\usepackage{comment}
\usepackage[mathscr]{eucal}
\usepackage[inline]{enumitem}
\usepackage{framed}
\usepackage{float}
\usepackage{fancyhdr}
\usepackage{graphicx}
\usepackage{lipsum}
\usepackage{mathtools}
\usepackage{multirow}
\usepackage[authoryear]{natbib}
\usepackage{proof}
\usepackage{ragged2e}
\usepackage{siunitx}
\usepackage{stmaryrd}
\usepackage[subrefformat=parens]{subcaption}
\usepackage{xcolor}
\usepackage{tikz}
\usepackage{tikz-qtree}
\usepackage{thmtools}
\usepackage{wrapfig}

\algnewcommand{\algorithmicgoto}{\textbf{goto}}
\algnewcommand{\Goto}[1]{\algorithmicgoto~\ref{#1}}

\algnewcommand{\algorithmicinitialize}{\textbf{initialize}}
\algnewcommand{\Initialize}{\algorithmicinitialize}

\algnewcommand{\algorithmicbreak}{\textbf{break}}
\algnewcommand{\Break}{\algorithmicbreak}
\algblockdefx[Elseb]{Elseb}{EndElseb}{\algorithmicelse}{\algorithmicend\ \algorithmicelse}
\makeatletter
\ifthenelse{\equal{\ALG@noend}{t}}%
  {\algtext*{EndElseb}}
  {}%
\makeatother

\usetikzlibrary{calc}
\usetikzlibrary{decorations.pathreplacing}
\usetikzlibrary{positioning}
\usetikzlibrary{shadings}
\usetikzlibrary{shadows}
\usetikzlibrary{shapes.arrows}
\usetikzlibrary{shapes}

\usepackage{hyperref}
\hypersetup{
  colorlinks=true,
  citecolor=black,
  urlcolor=magenta,
}
\declaretheorem[style=plain,numberwithin=section]{theorem}
\declaretheorem[style=plain,sibling=theorem]{corollary}
\declaretheorem[style=plain,sibling=theorem]{proposition}

\declaretheorem[style=definition,sibling=theorem]{definition}
\declaretheorem[style=definition,sibling=theorem]{example}
\declaretheorem[style=remark,sibling=theorem]{remark}

\allowdisplaybreaks

\newcommand{\set}[1]{\{#1\}}

\newcommand{\defas}{\coloneqq}
\newcommand{\asdef}{\eqqcolon}

\newcommand{\ceil}[1]{\lceil#1\rceil}

\newcommand{\Indicator}{\mathbf{1}}
\newcommand{\flip}{\textsc{Flip}}
\newcommand{\bb}{\mathbf{b}}

\newcommand{\Rationals}{\mathbb{Q}}
\newcommand{\Bools}{\mathbb{B}}
\newcommand{\Naturals}{\mathbb{N}}
\newcommand{\NumSys}[1]{\Bools_{#1}}

\newcommand{\distrib}[1]{\mathsf{#1}}
\newcommand{\bernoulli}{\distrib{Bernoulli}}
\newcommand{\uniform}{\distrib{Uniform}}

\begin{document}

\runningtitle{The Fast Loaded Dice Roller: A Near-Optimal Exact Sampler
    for Discrete Probability Distributions}

\twocolumn[
  \aistatstitle{
    The Fast Loaded Dice Roller: A Near-Optimal \\ Exact Sampler
    for Discrete Probability Distributions
  }
  \aistatsauthor{
    Feras A.~Saad
    \And Cameron E.~Freer
    \And Martin C.~Rinard
    \And Vikash K.~Mansinghka
  }
  \aistatsaddress{
    MIT EECS
    \And MIT BCS
    \And MIT EECS
    \And MIT BCS}
]

\begin{abstract}
This paper introduces a new algorithm for the
  fundamental problem of generating a random integer from a discrete
  probability distribution using a source of independent and unbiased
  random coin flips.
We prove that this algorithm, which we call the Fast Loaded Dice Roller (FLDR),
  is highly efficient in both space and time:
  \begin{enumerate*}[label=(\roman*)]
  \item the size of the sampler is guaranteed to be linear in the
    number of bits needed to encode the input distribution; and
  \item the expected number of bits of entropy it consumes per
    sample is at most 6 bits more than the
    information-theoretically optimal rate.
  \end{enumerate*}
We present fast implementations of the linear-time preprocessing
  and near-optimal sampling algorithms using unsigned integer
  arithmetic.
Empirical evaluations on a broad set of probability distributions
  establish that FLDR
  is 2x--10x faster in both preprocessing and sampling than
  multiple baseline algorithms,
  including the widely-used alias and interval samplers.
It also uses up to 10000x less space than the
  information-theoretically optimal sampler, at the expense of less
  than 1.5x runtime overhead.
\end{abstract}

\section{INTRODUCTION}
\label{sec:introduction}

The problem of generating a discrete random variable is as follows:
  given a probability distribution $p \defas (p_1, \dots, p_n)$ and
  access to a random source that outputs an independent stream of fair
  bits, return integer $i$ with probability $p_i$.
A classic theorem from~\citet{knuth1976} states that the most
  efficient sampler, in terms of the expected number of random bits
  consumed from the source, uses between $H(p)$ and $H(p) + 2$ bits
  in expectation, where $H(p) \defas \sum_{i=1}^{n}p_i\log(1/p_i)$ is the
  Shannon entropy of $p$.
This entropy-optimal sampler is obtained
  by building a decision tree using
  the binary expansions of the $p_i$.

Despite the fact that the \citeauthor{knuth1976} algorithm
  provides the most time-efficient sampler for any probability distribution,
  this paper shows that its construction can require
  exponentially larger space than the number of bits needed to encode
  the input instance $p$ and may thus be infeasible to
  construct in practice.
In light of this negative result, we aim to develop a sampling
  algorithm whose entropy consumption is close to the optimal rate and
  whose space scales polynomially.

This paper presents a new sampling algorithm where, instead
  of using an entropy-optimal sampler to simulate
  $(p_1,\dots,p_n)$ directly, we define a proposal distribution
  $(q_1, \dots, q_n, q_{n+1})$ on an extended domain whose probabilities
  $q_i$ are dyadic rationals that are ``close'' to the probabilities
  $p_i$ and then simulate the proposal with an
  entropy-optimal sampler followed by an accept/reject step.
We prove that this sampling algorithm,
  which we call the Fast Loaded Dice Roller (FLDR),
  is efficient in both
  space and time: its size scales linearly in
  the number of bits needed to encode the input instance $p$ and
  it consumes between $H(p)$ and $H(p) + 6$ bits in expectation, which
  is near the optimal rate and does not require exponential memory.

We present an implementation of FLDR using fast
  integer arithmetic and show empirically that it is
  2x--10x faster than several exact baseline samplers,
  and uses up to 10000x less space than the
  entropy-optimal sampler of~\citeauthor{knuth1976}.
To the best of our knowledge, this paper presents the first
  theoretical characterization and practical implementation of using
  entropy-optimal proposal distributions for accept-reject sampling, as
  well as benchmark measurements that highlight the space and runtime
  benefits of FLDR over multiple existing exact sampling algorithms.
A prototype implementation in C is released with the paper.

The remainder of this paper is structured as follows.
Section~\ref{sec:random-bit-model} formally introduces the random bit
  model of computation for studying the sampling
  algorithms used throughout the paper.
Section~\ref{sec:knuth-yao-complexity} establishes the worst-case
  exponential space of the entropy-optimal
  \citeauthor{knuth1976} sampler.
Section~\ref{sec:rejection-sampling} presents a systematic study of
  the space--time complexity of three common baseline rejection
  algorithms.
Section~\ref{sec:fast-loaded-dice-roller} presents FLDR and
  establishes its linear memory and near-optimal entropy consumption.
Section~\ref{sec:empirical-measurements} presents
  measurements of the preprocessing time, sampling time, and memory
  consumption of FLDR and demonstrates improvements over
  existing exact samplers.

\section{PRELIMINARIES}
\label{sec:random-bit-model}

\paragraph{Algebraic model}
Many algorithms for sampling discrete random
  variables~\citep{walker1977,vose1991,smith2002,bringmann2017}
  operate in a model of computation where the space--time complexity
  of both preprocessing and sampling are analyzed assuming a
  real RAM model~\citep{blum1998} (i.e., storing and arithmetically
  manipulating infinitely precise numbers can be done in constant
  time~\citep[Assumptions I, III]{devroye1986}).
Algorithms in this model apply a sequence of transformations to a
  uniform random variable $U \in [0,1]$, which forms the basic unit of
  randomness~\citep[Assumption II]{devroye1986}.
While often useful in practice, this model does not
  permit a rigorous study of either the complexity, entropy
  consumption, or sampling error of different samplers.
More specifically, real RAM sampling algorithms typically generate
  random variates which are only approximately distributed according
  to the target distribution when implemented on physically-existing
  machines due to limited numerical precision, e.g., IEEE
  double-precision floating-point~\citep{bringmann2013}.
This sampling
  error is challenging to quantify in
  practice~\citep{devroye1982,monahan1985}.
In addition, the real RAM model does not account for the complexity of drawing
  and manipulating the random variable $U$ from the underlying source
  (a single uniform random variate has the same amount of entropy
  as countably infinitely many such variates) and thus ignores a key design
  constraint for samplers.

\paragraph{Random bit model}
This paper focuses on exact sampling (i.e., with zero sampling error)
  in a word RAM model of computation where the basic unit of
  randomness is an independent, unbiased bit $B \in \set{0,1}$
  returned from a primitive operation \textsc{Flip}.
The random bit model is widely used, both in information
  theory~\citep{han1993}
  and in formal descriptions of sampling algorithms
  for discrete distributions
  that use finite precision arithmetic and random fair bits.
Examples include the
  uniform~\citep{lumbroso2013},
  discrete Gaussian~\citep{follath2014},
  geometric~\citep{bringmann2017},
  random graph~\citep{blanca2012},
  and general categorical~\citep{knuth1976,uyematsu2003} distributions.
The model has also been generalized to the setting of using
  a biased or non-i.i.d.\ source of coin
  flips for sampling~\citep{neumann1951,elias1972,blum1986,roche1991,peres1992,abrahams1996,pae2006,kozen2018}.

\paragraph{Problem Formulation}
  Given a list $(a_1, \dots, a_n)$ of $n$ positive integers which sum to
  $m$ and access to a stream of independent
  fair bits (i.e., $\flip$), sample integer $i$ with
  probability $a_i/m$ $(i=1,\dots,n)$.

Designing algorithms and data structures for this problem of ``dice rolling''
  has received widespread attention in
  the computer science literature; see~\citet{schwarz2011} for a survey.
We next describe a framework for describing the computational behavior
  of any sampling algorithm implemented in the random bit model.

\paragraph{Discrete distribution generating trees}
\citet{knuth1976} present a computational framework for expressing any
  sampling algorithm in the random bit model in terms of a
  (possibly infinite) rooted binary tree $T$, called a
  \textit{discrete distribution generating} (DDG) tree,
  which has the following properties:
  \begin{enumerate*}[label=(\roman*)]
      \item each internal node has exactly
        two children (i.e., $T$ is full); and
      \item each leaf node is labeled with one outcome from the set
        $\set{1,2,\dots,n}$.
  \end{enumerate*}
The algorithm is as follows: starting at the root, obtain a
  random bit $B \sim \flip$.
Proceed to the left child if $B=0$ and proceed to the right child if
  $B=1$. If the child node is a leaf, return the label assigned to
  that leaf and halt.
Otherwise, draw a new random bit $B$ and repeat the process.
For any node $x \in T$, let $l(x)$ denote its label and $d(x)$ its
  level (by convention, the root is at level 0 and all internal nodes
  are labeled $0$).
Since $\flip$ returns fair bits, the output probability distribution
  $(p_1, \dots, p_n)$ is
  \begin{align*}
  p_i \defas \mathbb{P}[T \mbox{ returns } i]
    = \sum\limits_{x \,\mid\, l(x) = i} 2^{-d(x)} && (i=1,\dots,n).
  \end{align*}
The number of coin flips $L_T$ used when simulating $T$
  is, in expectation, the average depth of the leaves, i.e.,
  \begin{align*}
  \mathbb{E}[L_T] = \sum_{x \,\mid\, l(x) > 0}d(x)2^{-d(x)}.
  \end{align*}
The operators $\mathbb{P}$ and $\mathbb{E}$ are defined over the
  sequence $\bb \in \set{0,1}^{\infty}$ of bits from the random
  source, finitely many of which are consumed during a halting
  execution (which occurs with probability one).
The following classic theorem establishes tight bounds on the minimal
  expected number of bits consumed by any sampling algorithm for a
  given distribution $p$, and provides an explicit construction
  of an optimal DDG tree.

\begin{theorem}[\citet{knuth1976}]
\label{thm:ddg-knuth-yao}
Let $p \defas (p_1, \dots, p_n)$, where $n > 1$.
Any sampling algorithm with DDG tree $T$ and output distribution
  $p$ whose expected number of input bits is minimal
  (among all trees $T'$ whose output distribution equals $p$)
  satisfies $H(p) \le \mathbb{E}[L_T] < H(p) + 2$.
These bounds are the tightest possible.
In addition, $T$ contains exactly 1 leaf node
  labeled $i$ at level $j$ if and only if $p_{ij} = 1$,
  where $(0.p_{i1}p_{i2}\dots)_2$ denotes the
  binary expansion of each $p_i$
  (which ends in $\bar{0}$ whenever $p_i$ is dyadic).
\end{theorem}

We now present examples of DDG trees.

\begin{example}
\label{example:ddg-dyadic}
Let $p \defas (1/2, 1/4, 1/4)$.
By Thm.~\ref{thm:ddg-knuth-yao}, an entropy-optimal DDG tree for $p$
can be constructed directly from the
  binary expansions of the $p_i$, where $p_{ij}$ corresponds
  to the $j$th bit in the binary expansion of $p_i$
  ($i = 1,2,3; j \ge 0$).
Since $p_1 = (0.10)_2$ and $p_2 = p_3 = (0.01)_2$ are all dyadic, the
  entropy-optimal tree has three levels, and the sampler always halts
  after consuming at most 2 bits.
Also shown is an entropy-suboptimal tree for $p$, which always
  halts after consuming at most 3 bits.
\begin{figure}[H]
\begin{subfigure}[b]{.5\linewidth}
\centering
  \begin{tikzpicture}
  \centering
  \tikzset{level distance=10pt}
  \tikzset{every tree node/.style={anchor= north}}
  \Tree [
    [ 3 2 ] 1
  ]
  \end{tikzpicture}
  \caption*{Optimal DDG tree}
\end{subfigure}\hfill
\begin{subfigure}[b]{.5\linewidth}
\centering
  \begin{tikzpicture}
  \centering
  \tikzset{level distance=10pt}
  \tikzset{every tree node/.style={anchor= north}}
  \Tree [
    [ 1 [ 2 3 ] ] [ 1 [ 3 2 ]  ]
  ]
  \end{tikzpicture}
  \caption*{Suboptimal DDG tree}
\end{subfigure}%
\end{figure}
\end{example}

\begin{example}
\label{example:ddg-rat}
Let $p \defas (3/10, 7/10)$.
Although $p_1$ and $p_2$ have infinite binary expansions, they are
  rational numbers which can be encoded using a
  finite prefix and a bar atop a finite repeating suffix;
  i.e., $p_1 = (0.0\overline{1001})_2, p_2 = (0.1\overline{0110})_2$.
While any DDG tree for $p$ has infinitely many
  levels, it can be finitely encoded by using back-edges
  (shown in red).
The entropy-optimal tree has five levels and a back-edge from level 4
  to level 1, corresponding to the binary expansions of the
  $p_i$, where the suffixes have four digits and prefixes have one
  digit.
\begin{figure}[H]
\centering
\begin{subfigure}[b]{.5\linewidth}
\centering
  \begin{tikzpicture}
  \centering
  \tikzstyle{branch}=[shape=coordinate]
  \tikzstyle{leaf}=[circle,draw=red,fill=red,inner sep=0pt]
  \tikzset{level distance=10pt}
  \tikzset{every tree node/.style={anchor= north}}
  \Tree [
    [.\node[branch](b1){};
      [ [ [ \edge[color=red];\node[leaf](r1){}; 1 ] 2 ] 2 ] 1 ]
    2
  ]
  \draw[->,red] (r1.west) to[bend left=80] ([xshift=-.1cm]b1.west);
  \end{tikzpicture}
  \caption*{Optimal DDG tree}
  \label{fig:ddg-rat-opt}
\end{subfigure}\hfill
\begin{subfigure}[b]{.45\linewidth}
\centering
  \begin{tikzpicture}
  \centering
  \tikzstyle{branch}=[shape=coordinate]
  \tikzstyle{leaf}=[circle,draw=red,fill=red,inner sep=0pt]
  \tikzset{level distance=10pt}
  \tikzset{every tree node/.style={anchor= north}}
  \Tree
  [.\node[branch](b1){};
    [
      \edge[color=red];\node[leaf](r3){};
      [
        \edge[color=red];\node[leaf](r2){};
        [ \edge[color=red];\node[leaf](r1){}; 2 ]
      ]
    ]
    [
      [ 2 1 ]
      2
    ]
  ]
  \draw[->,red] (r1.west)
      to[out=180, in=-90]
      ($(r3)-(.5,0)$)
      to[out=90, in=140]
      ([xshift=-.1cm]b1.west);
  \draw[->,red] (r2.west)
      to[out=180, in=-90]
      ($(r3)-(.30,0)$)
      to[out=90, in=140]
      ([xshift=-.1cm]b1.west);
  \draw[->,red] (r3.south)
    % to[bend left=90,]
      to[out=160, in=140]
      % ($(r3)-(.15,0)$)
      % to[out=90, in=140]
      ([xshift=-.1cm]b1.west);
  \end{tikzpicture}
  \caption*{Suboptimal DDG tree}
  \label{fig:ddg-rat-subopt}
\end{subfigure}%
\end{figure}
\end{example}

\begin{comment}
\begin{example}[\citet{knuth1976}]
\label{example:ddg-irrat}
Let $p \defas (1/\pi,\, 1/e,\, 1-1/\pi- 1/e)$, so that $p_1 =
  (0.010100\dots)_2$, $p_2= (0.0101111\dots)_2$, and $p_3 =
  (0.010100\dots)_2$.
%
As $p$ contains irrational probabilities the binary
  expansions have no repetitions.
%
Any DDG tree thus has an infinite number of levels and cannot be
  finitely encoded.
\begin{figure}[H]
\centering
\begin{subfigure}[b]{.45\linewidth}
  \begin{tikzpicture}
  \centering
  \tikzset{level distance=10pt}
  \tikzset{every tree node/.style={anchor= north}}
  \Tree [
    [ [ [ [ $\dots$ 2 ] 3 ] [ 1 2 ] ] 3 ]
    [ 2 1 ]
  ]
  \end{tikzpicture}
  \caption{Optimal DDG tree}
\end{subfigure}
\end{figure}
\end{example}
\end{comment}

% The above figures are not associated with caption,
% so renumber them.
\setcounter{figure}{0}

\begin{definition}[Depth of a DDG tree]
\label{def:ddg-depth}
Let $T$ be a DDG tree over $\set{1,\dots,n}$ with output distribution
  $(p_1, \dots, p_n)$, where each $p_i \in \Rationals$.
We say that $T$ has depth $k$ if the longest path from the root node
  to any leaf node in the shortest finite tree encoding of $T$ (using
  back-edges, as in Example~\ref{example:ddg-dyadic}) consists of $k$
  edges.
\end{definition}

In this paper, we do not consider distributions with irrational entries,
  as their DDG trees are infinite and cannot be finitely encoded.
Thm.~\ref{thm:ddg-knuth-yao} settles the
  problem of constructing the most ``efficient'' sampler for a target
  distribution, when efficiency is measured by the expected number of
  bits consumed.

However, designing an entropy-efficient sampler that is also
  space-efficient remains an open problem.
In particular, as we show in Section~\ref{sec:knuth-yao-complexity},
  the size of the optimal DDG tree $T$ is exponentially larger
  than the number of bits needed to encode $p$ and is therefore often
  infeasible to construct in practice.
\citet{knuth1976} allude to this issue,
  saying ``most of the algorithms which
  achieve these optimum bounds are very complex, requiring a
  tremendous amount of space''.

\section{COMPLEXITY OF ENTROPY- OPTIMAL SAMPLING}
\label{sec:knuth-yao-complexity}

This section recounts background results from \citet[Section~3]{saad2020popl}
  about the class of
  entropy-optimal samplers given in Thm.~\ref{thm:ddg-knuth-yao}.
These results establish the worst-case exponential space of
  entropy-sampling and formally motivate the need for space-efficient
  and near-optimal samplers developed in
  Section~\ref{sec:fast-loaded-dice-roller}.
For completeness, the proofs are presented in Appendix~\ref{appx:proofs}.

For entropy-optimal DDG trees that have depth $k \ge 1$
  (Definition~\ref{def:ddg-depth}), the output probabilities are
  described by a fixed-point $k$-bit number.
The fixed-point $k$-bit numbers $x$ are those such that for some
  integer $l$ satisfying $0 \le l \le k$, there is an element $(x_1,
  \dots, x_k) \in \set{0,1}^{l} \times \set{0,1}^{k-l}$, where the
  first $l$ bits correspond to a finite prefix and the final $k-l$
  bits correspond to an infinitely repeating suffix, i.e., $x = (0.x_1
  \ldots x_l \overline{x_{l+1} \ldots x_{k}})_2$.
Write $\NumSys{kl}$ for the set of rationals in $[0,1]$ describable in
  this way.
\begin{proposition}
\label{prop:numsys-bases}
For integers $k$ and $l$ with $0 \le l \le k$,
  define $Z_{kl} \defas 2^k - 2^{l}\Indicator_{l < k}$.
Then
  \begin{align*}
  \NumSys{kl} =
  \left\lbrace
    \frac{0}{Z_{kl}},
    \frac{1}{Z_{kl}},
    \dots,
    \frac{Z_{kl}-1}{Z_{kl}},
    \frac{Z_{kl}}{Z_{kl}}\Indicator_{l < k}
  \right\rbrace.
  \label{eq:wqa}
  \end{align*}
\end{proposition}

The next result establishes that the number systems $\NumSys{kl}$
  ($k\,{\in}\,\Naturals$, $0\,{\le}\,l\,{\le}\,k$) from Prop.~\ref{prop:numsys-bases}
  describe the output probabilities of
  optimal DDG trees with depth-$k$.

\begin{theorem}
\label{thm:k-bit-bases}
Let $T$ be an entropy-optimal DDG tree with a non-degenerate
  output distribution $(p_i)_{i=1}^{n}$
  for $n > 1$.
The depth of $T$ is the smallest integer $k$ such that there exists
  an integer $l \in \set{0,\dots,k}$ for which all the $p_i$
  are integer multiples of $1/Z_{kl}$ (hence in $\NumSys{kl}$).
\end{theorem}

\begin{corollary}
Every back-edge in an entropy-optimal depth-$k$ DDG tree
  originates at level $k-1$ and ends at the same level $l$,
  where $0 \le l < k-1$.
\end{corollary}
The next result, Thm.~\ref{thm:depth-perfect-sampling}, implies
  that an entropy-optimal DDG tree for a coin with weight $1/m$
  has depth at most $m-1$.
Thm.~\ref{thm:depth-perfect-sampling-prime} shows that this bound
  is tight for many $m$, and
Rem.~\ref{remark:Artin} notes that it is likely tight for infinitely
  many $m$.

\begin{theorem}
\label{thm:depth-perfect-sampling}
Suppose $p$ is defined by $p_i = a_i/m$ $(i=1,\dots,n)$,
  where $\sum_{i=1}^{n}a_i=m$.
The depth of any entropy-optimal sampler for
  $p$ is at most $m-1$.
\end{theorem}

\begin{theorem}
\label{thm:depth-perfect-sampling-prime}
Let $p$ be as in Thm.~\ref{thm:depth-perfect-sampling}.
If $m$ is prime and 2 is a primitive root modulo $m$,
  then the depth of an entropy-optimal DDG tree for
  $p$ is $m-1$.
\end{theorem}

\begin{remark}
\label{remark:Artin}
The bound in Thm.~\ref{thm:depth-perfect-sampling} is likely
  the tightest possible for infinitely many $m$.
Assuming Artin's conjecture, there are infinitely many primes $m$ for
  which $2$ is a primitive root, which by
  Thm.~\ref{thm:depth-perfect-sampling-prime} implies any
  entropy-optimal DDG tree has depth $m$.
\end{remark}

Holding $n$ fixed,
  the tight upper bound $m$ on the depth of an entropy-optimal DDG tree
  for any distribution having an entry $1/m$ is thus exponentially larger (in $m$) than
  the $n\log(m)$ bits needed to encode the input instance
  (each of $a_1, \dots, a_n$ requires a word of size at least $\log(m)$ bits).
Fig.~\ref{fig:ddg-depth-bernoulli} shows a plot of the scaling
  characteristics from Thm.~\ref{thm:depth-perfect-sampling}
  and provides evidence for the tightness conjectured in
  Rem.~\ref{remark:Artin}.

\section{REJECTION SAMPLING}
\label{sec:rejection-sampling}

We now present several alternative algorithms for exact sampling based
  on the rejection method~\citep[II.3]{devroye1986}, which lead to the
  Fast Loaded Dice Roller presented in Section~\ref{sec:fast-loaded-dice-roller}.
Rejection sampling operates as follows:
given a \textit{target distribution} $p \defas (p_1, \dots, p_n)$ and
  \textit{proposal distribution} $q \defas (q_1, \dots, q_l)$ (with $n \le l$),
  first find a \textit{rejection bound} $A > 0$
  such that $p_i \le Aq_i$ ($i=1,\dots,n$).
Next, sample $Y \sim q$
  and flip a coin with weight $p_Y / (Aq_Y)$
  (where $p_{n+1} = \ldots = p_l = 0$):
  if the outcome is heads accept $Y$, otherwise repeat.
The probability of halting in any given round is:
  \begin{align*}
  \Pr\left[\bernoulli\left(\frac{p_Y}{Aq_Y}\right) = 1\right]
    = \sum_{i=1}^{n} p_i/(Aq_i) \, q_i=1/A.
  \end{align*}
The number of trials thus follows a geometric distribution with
  rate $1/A$, whose mean is $A$.
We next review common implementations of random-bit rejection samplers
  and their space--time characteristics.
All algorithms take $n$ positive
  integers $(a_1, \dots, a_n)$ and the sum $m$ as input, and return
  $i$ with probability $a_i/m$.

\paragraph{Uniform Proposal}
Consider the uniform proposal distribution
  $q \defas (1/n, \dots, 1/n)$.
Set
  $D \defas \max_i(a_i)$ and set
  $A \defas Dn/m$, which gives a tight rejection bound since
  $p_i \le \max_i(p_i) = D/m = (Dn/m) (1/n) = A q_i$,
  so that $i$ is accepted with probability
  $a_i / D$ \linebreak ($i=1,\dots,n$).
Alg.~\ref{alg:rejection-uniform}
  presents an implementation where
  \begin{enumerate*}[label=(\roman*)]
    \item simulating the uniform proposal
    (line~\ref{algline:rejection-uniform-fdr}), and
    \item accepting/rejecting the proposed sample
    (line~\ref{algline:rejection-uniform-bernoulli}),
  \end{enumerate*}
  are both achieved using the two entropy-optimal samplers
  in~\citet{lumbroso2013} for uniform and Bernoulli generation.
The only extra storage needed by Alg.~\ref{alg:rejection-uniform}
  is in computing the maximum $D$ during preprocessing
  (line~\ref{algline:rejection-uniform-preprocess}).
For runtime, $A\,{=}\,nD$ trials occur on average; each trial uses
  $\log{n}$ bits for sampling $\uniform(n)$ and 2 bits for
  sampling $\bernoulli(a_i/D)$ on average.
The entropy is therefore order
  $n(m-n)\log{n} \ge n\log{n} \gg \log n$ bits.
Thus, despite its excellent space and
  preprocessing characteristics, the method can be exponentially
  wasteful of bits.

%!TEX root=../paper.tex
\begin{algorithm}[H]
\caption{\footnotesize Rejection sampler (uniform)}
\label{alg:rejection-uniform}
\begin{algorithmic}[1]
\footnotesize
% \Require Positive integers $(a_1, \dots, a_n)$, $m \defas \sum_{i=1}^{n}a_i$.
% \Ensure Random integer $i$ with probability $a_i / m$.
%
\Statex \texttt{// PREPROCESS}
\State Let $D \gets \max(a_1, \dots, a_n)$;
  \label{algline:rejection-uniform-preprocess}
\Statex \texttt{// SAMPLE}
\While{true}
  \label{algline:rejection-uniform-while}
  \State $i \sim \textsc{FastDiceRoller}(n)$;\,(\citet[p.\,4]{lumbroso2013})
    \label{algline:rejection-uniform-fdr}
  \State $x \sim \textsc{Bernoulli}(a_i, D)$; (\citet[p.\,21]{lumbroso2013})
    \label{algline:rejection-uniform-bernoulli}
  \If{$(x = 1)$} \Return $i$; \EndIf
\EndWhile
\end{algorithmic}
\end{algorithm}

\paragraph{Dyadic Proposal}
Consider the following proposal distribution.
Let $k \in \Naturals$ be such that
  $2^{k-1} < m \le 2^{k}$
  (i.e., $k-1 < \log(m) \le k$ so that $k = \ceil{\log{m}}$) and set
  \begin{equation}
  \label{eq:dyadic-proposal}
  q \defas (a_1/2^k, \dots, a_n/2^k, 1-m/2^k).
  \end{equation}
The tightest rejection bound
  $A = 2^k/m$, since $p_i = a_i/m = (2^k/m)a_i/2^k = Aq_i$
  ($i=1,\dots,n$) and $p_{n+1} = 0 \le (2^k-m)/2^k = q_{n+1}$.
Thus, $i$ is always accepted when $1 \le i \le n$ and always rejected
  when $i=n+1$.

{\itshape Lookup-table Implementation}.
\citet{devroye1986} implements the rejection
  sampler with proposal Eq.~\eqref{eq:dyadic-proposal}
  using a length-$m$ lookup table $T$, which has exactly
  $a_i$ elements labeled $i$ ($i=1,\dots,n$), shown in
  Alg.~\ref{alg:rejection-dyadic-lookup-table}.
The sampler draws $k$ random bits $(b_1, \dots, b_k)$,
  forms an integer $W \defas \sum_{i=1}^{k}b_i2^{i-1}$, and returns
  $T[W]$ if $0 \le W < m$ or repeats if
  $m \le W \le 2^k-1$.
For fixed $n$, the $m\log{m}$ space required by $T$ is
  exponentially larger (in $m$) than the $n\log{m}$ bits
  needed to encode the input.
Further, the number of bits per trial is always $k$,
  so $k2^k/m \ge k \approx \log{m}$ bits are used on average, which
  (whenever $n \ll m$)
  can be much higher than the optimal rate, which is at most $H(p) + 2 \le \log{n} + 2$.
%
%!TEX root=../paper.tex
\begin{algorithm}[H]
\caption{\footnotesize Rejection sampler (dyadic + lookup table)}
\label{alg:rejection-dyadic-lookup-table}
\begin{algorithmic}[1]
\footnotesize
% \Require Positive integers $(a_1, \dots, a_n)$, $m \defas \sum_{i=1}^{n}a_i$.
% \Ensure Random integer $i$ with probability $a_i / m$.
%
\Statex \texttt{// PREPROCESS}
\State Let $k \gets \ceil{\log{m}}$;
\State Make size-$m$ table $T$ with $a_i$ entries $i$ $(i=1,\dots,n)$;
\Statex \texttt{// SAMPLE}
\While{true}
    \State Draw $k$ bits, forming integer $W \in \set{0,\dots,2^{k-1}}$;
    \If{$(W < m)$} \Return $T[W]$;
        \label{algline:lookup-table-lookup}\EndIf
\EndWhile
\end{algorithmic}
\end{algorithm}

{\itshape Binary Search Implementation}.
The exponential memory of the lookup table
  in Alg.~\ref{alg:rejection-dyadic-lookup-table}
  can be eliminated by inversion sampling the proposal Eq.~\eqref{eq:dyadic-proposal}
  using binary search on the cumulative frequencies, as shown in
  Alg.~\ref{alg:rejection-dyadic-binary-search}.
%
% This algorithm corresponds to inversion sampling the dyadic proposal
%   Eq.~\ref{eq:dyadic-proposal}.
%
% More specifically, we form a length-$(n+1)$ array $T$ where
%   $T[0] \defas 0$ and $T[j] \defas \sum_{i=1}^{j}a_i$ $(j=1,\dots,n)$.
% %
% If a random $k$-bit number $W < m$, then the index $j$ $(j=1,\dots,n)$ such
%   that $T[j-1] \le W < T[j]$ is returned (there are precisely $a_j$
%   values that satisfy this condition), and the procedure is
%   repeated otherwise, which corresponds to inversion sampling the proposal
%   (Alg.~\ref{alg:rejection-dyadic-binary-search}).
%
This algorithm consumes the same number of bits $k$ as
  Alg.~\ref{alg:rejection-dyadic-lookup-table}.
Its exponential improvement in space
  from $m\log{m}$ to $n\log{m}$
  introduces a logarithmic runtime factor
  the inner loop of the sampler, i.e.,
  line~\ref{algline:binary-search-lookup} of
  Alg.~\ref{alg:rejection-dyadic-binary-search}
  sometimes uses $\Omega(\log{n})$
  time as opposed to the constant indexing time from
  line~\ref{algline:lookup-table-lookup} of
  Alg.~\ref{alg:rejection-dyadic-lookup-table},
  representing a typical space--runtime tradeoff.

%!TEX root=../paper.tex
\begin{algorithm}[H]
\caption{\footnotesize Rejection sampler (dyadic + binary search)}
\label{alg:rejection-dyadic-binary-search}
\begin{algorithmic}[1]
\footnotesize
% \Require Positive integers $(a_1, \dots, a_n)$, $m \defas \sum_{i=1}^{n}a_i$.
% \Ensure Random integer $i$ with probability $a_i / m$.
%
\Statex \texttt{// PREPROCESS}
\State Let $k \gets \ceil{\log{m}}$;
\State Define array $T$, where $T[j] \defas \sum_{i=1}^{j}a_i$ $(j=1,\dots,n)$
\Statex \texttt{// SAMPLE}
\While{true}
    \State Draw $k$ bits, forming integer $W \in \set{0,\dots,2^{k-1}}$;
    \If{$(W < m)$} \Return $\min\set{j \mid W < T[j]}$;
        \label{algline:binary-search-lookup} \EndIf
\EndWhile
\end{algorithmic}
\end{algorithm}

\section{FAST LOADED DICE ROLLER}
\label{sec:fast-loaded-dice-roller}

Section~\ref{sec:rejection-sampling} shows that for
  rejection sampling using the dyadic proposal
  Eq.~\eqref{eq:dyadic-proposal}, a lookup table requires
  exponential memory and constant lookup time, whereas binary search
  uses linear memory but $\log{n}$ lookup time.
Moreover, these methods use $k$ bits/sample, which is highly
  wasteful for low-entropy distributions.
The key idea of the Fast Loaded Dice Roller (FLDR) presented in this
  section is to eliminate these
  memory, runtime, and entropy inefficiencies by simulating the
  proposal distribution $q$
  using an entropy-optimal sampler.

%!TEX root=../paper.tex
\begin{algorithm}[H]
\footnotesize
\caption{\footnotesize Fast Loaded Dice Roller (sketch)}
\label{alg:fldr-sketch}
\begin{enumerate}[itemsep=0mm]
\footnotesize
  \item Let $k \defas \ceil{\log{m}}$ and define the proposal distribution
    $q \defas (a_1/2^k, \dots, a_n/2^k, 1-m/2^k).$
    \label{item:phorometer}

  \item Simulate $X \sim q$, using an
    entropy-optimal sampler as described in Thm.~\ref{thm:ddg-knuth-yao}.
    \label{item:Larix}

  \item If $X \le n$, then return $X$, else go to Step~\ref{item:Larix}.
    \label{item:blastosphere}
\end{enumerate}
\end{algorithm}
% \vspace{-.5cm}

\vspace{-.5cm}
%!TEX root=../paper.tex
\begin{figure}[H]
\centering
\begin{subfigure}[b]{.5\linewidth}
\centering
  \begin{tikzpicture}
  \centering
  \tikzstyle{branch}=[shape=coordinate]
  \tikzstyle{leaf}=[circle,draw=red,fill=red,inner sep=0pt]
  \tikzset{level distance=10pt}
  \tikzset{every tree node/.style={anchor= north}}
  \Tree
    [.\node[branch](b1){};
      [
        [
          [
            \edge[color=red];\node[leaf](r1){};
            1 ]
          1 ]
        4 ]
      4 ]
  \draw[->,red] (r1.west) to[bend left=80] ([xshift=-.1cm]b1.west);
  \end{tikzpicture}
  \captionsetup{skip=2pt}
  \caption{Optimal DDG tree}
  \label{fig:ddg-compare-ky}
\end{subfigure}\hfill
\begin{subfigure}[b]{.45\linewidth}
\centering
  \begin{tikzpicture}
  \centering
  \tikzstyle{branch}=[shape=coordinate]
  \tikzstyle{leaf}=[circle,draw=red,fill=red,inner sep=0pt]
  \tikzset{level distance=10pt}
  \tikzset{every tree node/.style={anchor= north}}
  \Tree
    [.\node[branch](b1){};
      [
        \edge[color=red];\node[leaf](r2){};
        [
          \edge[color=red];\node[leaf](r1){};
          1
        ]
      ]
      4
    ]
  \draw[->,red] (r1.west)
      to[out=180, in=-90]
      ($(r2)-(.35,0)$)
      to[out=90, in=140]
      ([xshift=-.1cm]b1.west);
  \draw[->,red] (r2.south)
      to[out=160, in=140]
      ([xshift=-.1cm]b1.west);
  \end{tikzpicture}
  \captionsetup{skip=2pt}
  \caption{FLDR DDG tree}
  \label{fig:ddg-compare-fldr}
\end{subfigure}%
\caption{Comparison of DDG trees for $p = (1/5, 4/5)$.}
\label{fig:ddg-compare}
% \vspace{-.5cm}
\end{figure}
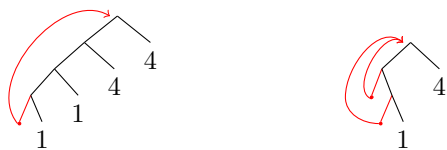

Fig.~\ref{fig:ddg-compare} shows a comparison of an entropy-optimal
  DDG tree and a FLDR DDG tree.
We next establish the linear space and near-optimal
  entropy of Alg.~\ref{alg:fldr-sketch}.

\begin{theorem}
\label{thm:rejection-dyadic-ddg-depth}
The DDG tree $T$ of FLDR in
  Alg.~\ref{alg:fldr-sketch} has at most
  $2(n+1)\ceil{\log{m}}$ nodes.
\end{theorem}

\begin{proof}
Suppose the DDG tree $T_q$ of the entropy-optimal sampler for $q$ in
  Step~\ref{item:Larix} of
  Alg.~\ref{alg:fldr-sketch} has $N$ total nodes,
  $s < N$ leaf nodes, and depth $k$.
Since $T_q$ is a full binary tree it has $N-1$ edges.
Moreover, the root has degree two, the $s$ leaves have degree one, and
  the $N-s-1$ internal nodes have degree three.
Equating the degrees and solving
  $2(N-1) = 2 + s + \hbox{$3(N-s-1)$}$ gives $N = 2s-1$.
Next, since $q$ is a dyadic distribution over $\set{1,\dots,n+1}$ with
  base $\ceil{\log{m}}$, $T_q$ has depth $k = \ceil{\log{m}}$
  (Thm.~\ref{thm:k-bit-bases}).
From the entropy-optimality of the depth-$k$ tree $T_q$ over
  $\set{1,\dots,n+1}$, we have $s \le (n+1)k$,
  since each of the $k$ levels has at
  most 1 leaf node labeled $i$ $(i=1,\dots,n+1)$
  (Thm.~\ref{thm:ddg-knuth-yao}).
Thus $N=2s-1 \le 2(n+1)k-1 \le 2(n+1)\ceil{\log{m}}$.
Finally, the DDG tree $T$ of FLDR is identical to $T_q$, except
  for additional back-edges from each leaf node labeled $n+1$ to the
  root (i.e., the rejection branch when $X=n+1$ in
  Step~\ref{item:blastosphere}).
\end{proof}

\begin{theorem}
\label{thm:rejection-dyadic-ddg-entropy}
The DDG tree $T$ of FLDR in
  Alg.~\ref{alg:fldr-sketch} satisfies
  \begin{align}
  0 \le \mathbb{E}[L_{T}] - H(p) < 6.
  \label{eq:Winebrennerian}
  \end{align}
\end{theorem}

\begin{proof}
Let $T_q$ be an entropy-optimal DDG tree for the proposal distribution
  $q$ defined in Step~\ref{item:phorometer},
  so that $\mathbb{E}[L_{T_q}] = H(q) + t_q$ for
  some $t_q$ satisfying $0 \le t_q < 2$ (by Thm.~\ref{thm:ddg-knuth-yao}).
Since the expected number of trials of
  Alg.~\ref{alg:fldr-sketch} is $2^k/m$ and the
  number of trials is independent of the bits consumed in
  each round, we have
  $\mathbb{E}[L_{T}] = (2^k/m) \mathbb{E}[L_{T_q}]$.

If $m = 2^k$ then $p = q$, and we have $\mathbb{E}[L_{T}] - H(p) = t_q$,
  so Eq.~\eqref{eq:Winebrennerian} holds.
Now suppose $m < 2^k$.
Then
  \begin{align}
  &\mathbb{E}[L_{T}] - H(p) \notag \\
  &= (2^k/m)(H(q) + t_q) - H(p) \notag \\
  &= (2^k/m)H(q) - H(p) + 2^k t_q/m \notag \\
  &= \begin{aligned}[t]
      & 2^k/m \bigl[\textstyle\sum_{i=1}^{n} a_i/2^k \log(2^k/a_i) \\
      &\quad\qquad+ (2^k-m)/2^k\log(2^k/(2^k-m))\bigr]\\
      &- \textstyle\sum_{i=1}^{n}a_i/m\log(m/a_i)
      + 2^kt_q/m \\
  \end{aligned} \notag \\
  &= \begin{aligned}[t]
    & \textstyle\sum_{i=1}^{n} a_i/m[\log(2^k/a_i) - \log(m/a_i)] \\
      &+ (2^k-m)/m\log(2^k/(2^k-m))
      + 2^kt_q/m \\
  \end{aligned} \notag \\
  &= \begin{aligned}[t]
    & \log(2^k/m) + (2^k-m)/m\log(2^k/(2^k-m)) \\
    & + 2^kt_q/m.
  \end{aligned} \label{eq:mushroomlike}
  \end{align}
We now bound Eq.~\eqref{eq:mushroomlike}
  under our restriction $2^{k-1} < m < 2^{k}$.
All three terms are monotonically decreasing in
  $m \in \set{2^{k-1},\dots,2^k-1}$,
  hence maximized when $m = 2^{k-1}+1$,
  achieving a value less than that for $m = 2^{k-1}$.
Hence the first term is less than
  $\log(2^k/2^{k-1}) = 1$,
  the second term less than
\begin{align*}
\frac{(2^k-2^{k-1})}{2^{k-1}}
        \log\left(\frac{2^k}{2^k-(2^{k-1})}\right)
        %\\
%    &\hspace*{80pt}
    = \log\left(\frac{2^k}{2^{k-1}}\right) = 1,
\end{align*}
and the third term less than $2t_q < 4$.
All three terms are positive, thus establishing bound
  Eq.~\eqref{eq:Winebrennerian}.
\end{proof}

Thms.~\ref{thm:rejection-dyadic-ddg-depth}
  and~\ref{thm:rejection-dyadic-ddg-entropy} together imply that
  Alg.~\ref{alg:fldr-sketch} uses $O(n\log{m})$
  space on a size $n\log{m}$ input instance and guarantees an
  entropy gap of at most $6$ bits sample, for any target distribution $p$.
Fig.~\ref{fig:ddg-depth-bernoulli} compares the asymptotic scaling of
  the size of the FLDR DDG tree from
  Thm.~\ref{thm:rejection-dyadic-ddg-depth} with that of the
  entropy-optimal sampler,
  and Fig.~\ref{fig:entropy-gap}
  decomposes the entropy gap from
  Thm.~\ref{thm:rejection-dyadic-ddg-entropy} according to the
  three terms in Eq.~\eqref{eq:mushroomlike}.

Alg.~\ref{alg:fldr} provides one of many possible implementations of FLDR
  (sketched in Alg.~\ref{alg:fldr-sketch}) that uses unsigned integer
  arithmetic to preprocess and sample an encoding of the underlying
  DDG tree.
This algorithm uses two data structures to eliminate the $O(n)$ inner-loop of the
  DDG tree sampler in \citet[Alg.~1]{roy2013} (at the
  expense of more memory), where array $h$ stores the number of leaf nodes at
  each level and matrix $H$ stores their labels in increasing order.
(A sparse matrix can often be used for $H$, as most of its entries are zero.)
Alternative DDG tree preprocessing and sampling algorithms that operate
  on an explicit tree data structure can be found
  in~\citet[Section~5]{saad2020popl}.

%!TEX root=../paper.tex
\begin{algorithm}[h]
\caption{\footnotesize
  Implementation of the Fast Loaded Dice Roller
  using unsigned integer arithmetic}
\label{alg:fldr}
\begin{algorithmic}[1]
\footnotesize
\Require Positive integers $(a_1, \dots, a_n)$, $m \defas \sum_{i=1}^{n}a_i$.
\Ensure Random integer $i$ with probability $a_i / m$.
\Statex \texttt{// PREPROCESS}
\State $k \gets \ceil{\log(m)}$;
  \label{algline:fldr-preprocess-start}
\State $a_{n+1} \gets 2^{k} - m$;
\State \Initialize\ $h\ \mathbf{int}[k]$;
\State \Initialize\ $H\ \mathbf{int}[n+1][k]$;
% \State \Initialize\ $\bH\,{\in}\,\set{0,\dots,n+1}^{(n+1)\times{k}}$;
\For{$j=0, \dots, k-1$}
    \State $d \gets 0$;
  \For{$i=1,\dots,n+1$}
    \State $\mbox{\bfseries bool }
      w \gets (a_i >> (k-1)-j))\; \mathsf{\&}\; 1$;
    \State $h[j] \gets h[j] + w$;
      \If{$w$}
        \State $H[d,j] \gets i$;
        \State $d \gets d + 1$;
    \EndIf
  \EndFor
\EndFor \label{algline:fldr-preprocess-end}
\Statex \texttt{// SAMPLE}
\State $d \gets 0, c \gets 0$; \label{algline:loop-init}
\While{true}
  \State $b \sim \textsc{Flip}()$;
  \State $d \gets 2\cdot d + (1-b)$;
  \If{$d < h[c]$}
    \If{$H[d,c] \le n$}
      \State \Return $H[d,c]$;
    \Else\; \{$d \gets 0$; $c \gets 0$;\}
      % \State \Goto{algline:loop-init};
    \EndIf
  \Else\; \{$d \gets d - h[c]$; $c \gets c+1$;\}
  \EndIf
\EndWhile \label{algline:loop-end}
\end{algorithmic}
\end{algorithm}

%!TEX root=../paper.tex
\begin{figure}[t]
\centering
\includegraphics[width=\linewidth]{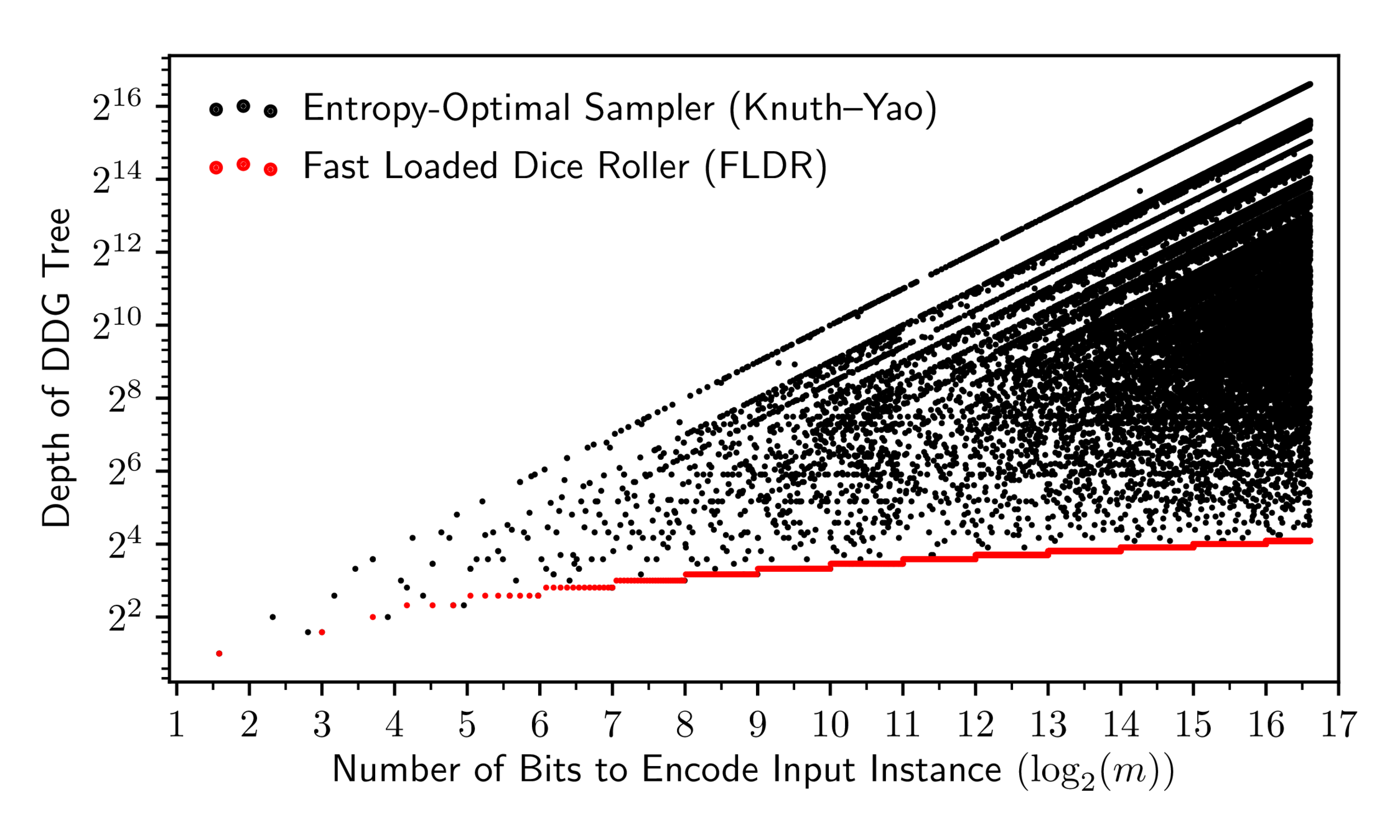}
\captionsetup{size=footnotesize, skip=0pt}
\caption{
  Depth of DDG tree for a distribution having an entry $1/m$,
    using the \citeauthor{knuth1976} entropy-optimal sampler (black) and
    FLDR (red) for $m = 3, \dots, 10^5$ (computed analytically).
  The y-axis is on a logarithmic scale:
    the entropy-optimal sampler scales exponentially
    (Thm.~\ref{thm:depth-perfect-sampling-prime}) and
    FLDR scales linearly
    (Thm.~\ref{thm:rejection-dyadic-ddg-depth}).
}
\label{fig:ddg-depth-bernoulli}
\end{figure}%
\vspace{-.5cm}%
\begin{figure}[t]
\includegraphics[width=\linewidth]{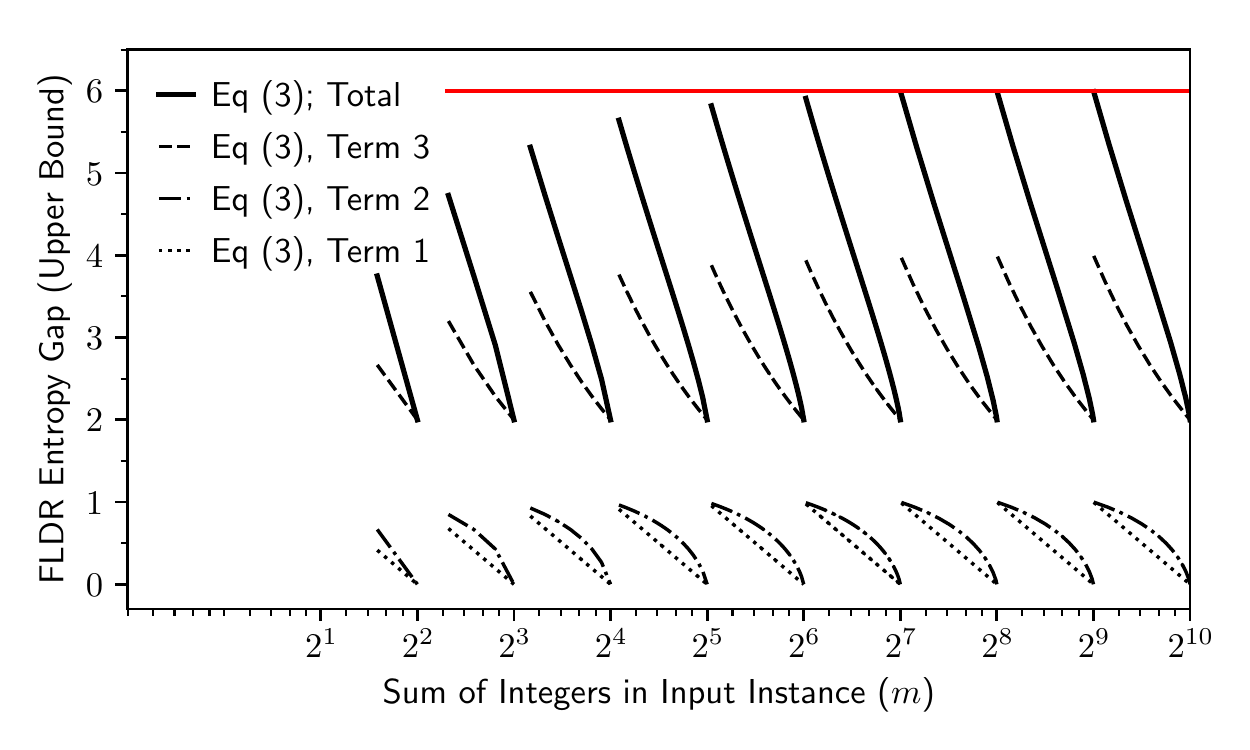}
\captionsetup{size=footnotesize, skip=0pt}
\caption{
  Plot of the three terms in Eq.~\eqref{eq:mushroomlike} in
    the entropy gap (y-axis) from
    Thm.~\ref{thm:rejection-dyadic-ddg-entropy},
    for varying $m$ (x-axis).
  %
  % The first term (dotted) is the additional entropy (over the domain
  %   $\set{1,\dots,n}$ of the target distribution $p$) that results from
  %   simulating the proposal $q$ an average of $2^k/m$ times.
  % %
  % The second term (dash-dot) is the additional entropy from the
  %   ``reject'' outcome $\set{n+1}$ of $q$.
  % %
  % The third term (dashed) is the 2 bit gap from
  %   Thm.~\ref{thm:ddg-knuth-yao}, weighted by $2^k/m$.
  % %
  % As $m \to 2^k$, the first two terms vanish
  % and the third term approaches $2$.
  % %
  % The solid black line is the sum of the three terms,
  % which is upper bounded by 6 (red line).
}
\label{fig:entropy-gap}
% \vspace{-.2cm}
\end{figure}

%!TEX root=../paper.tex
\begin{figure*}[t]
\centering
\begin{subfigure}{.5\linewidth}
\includegraphics[width=\linewidth]{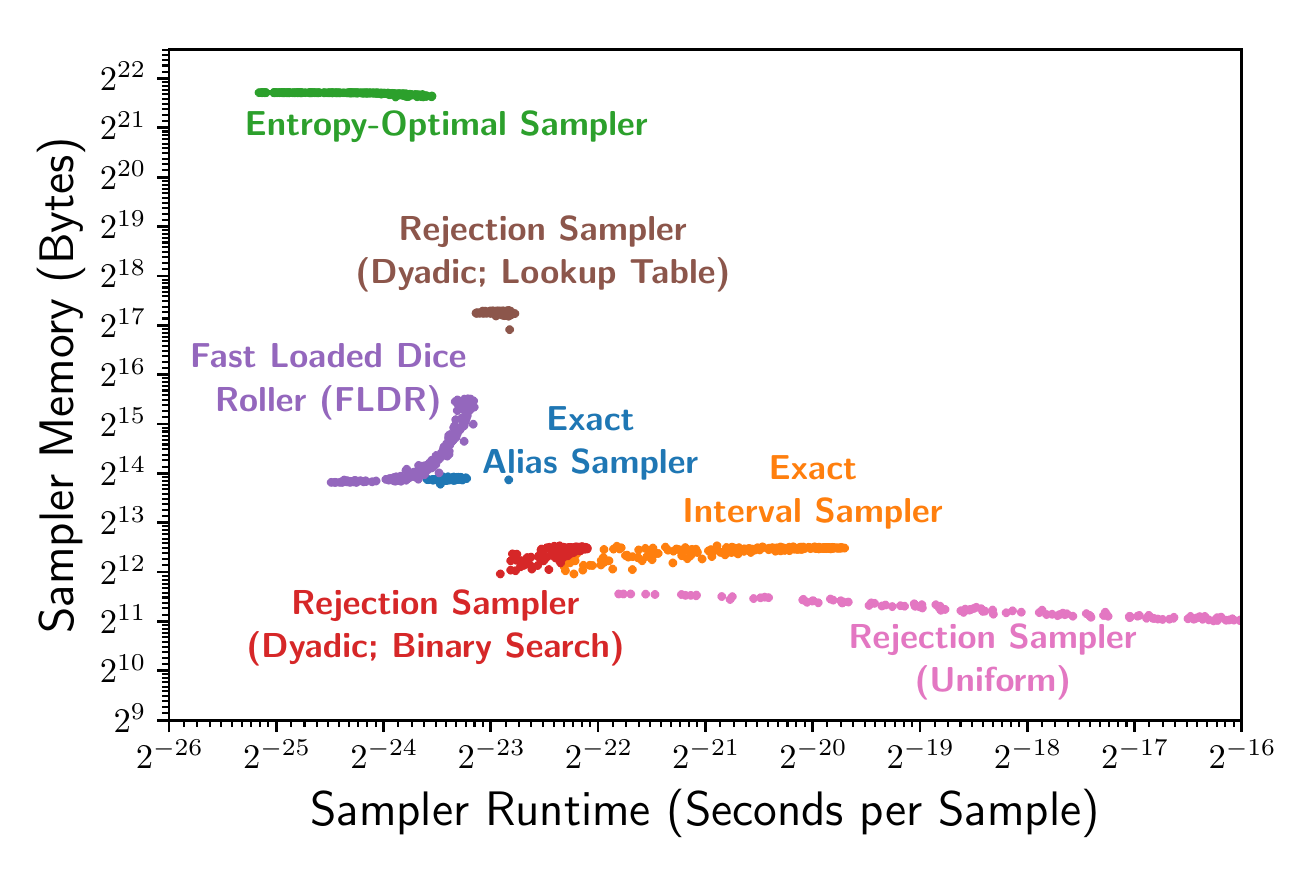}
\captionsetup{size=footnotesize, skip=0pt}
\caption{}
\label{fig:memory-rate-100}
\end{subfigure}%
\begin{subfigure}{.5\linewidth}
\includegraphics[width=\linewidth]{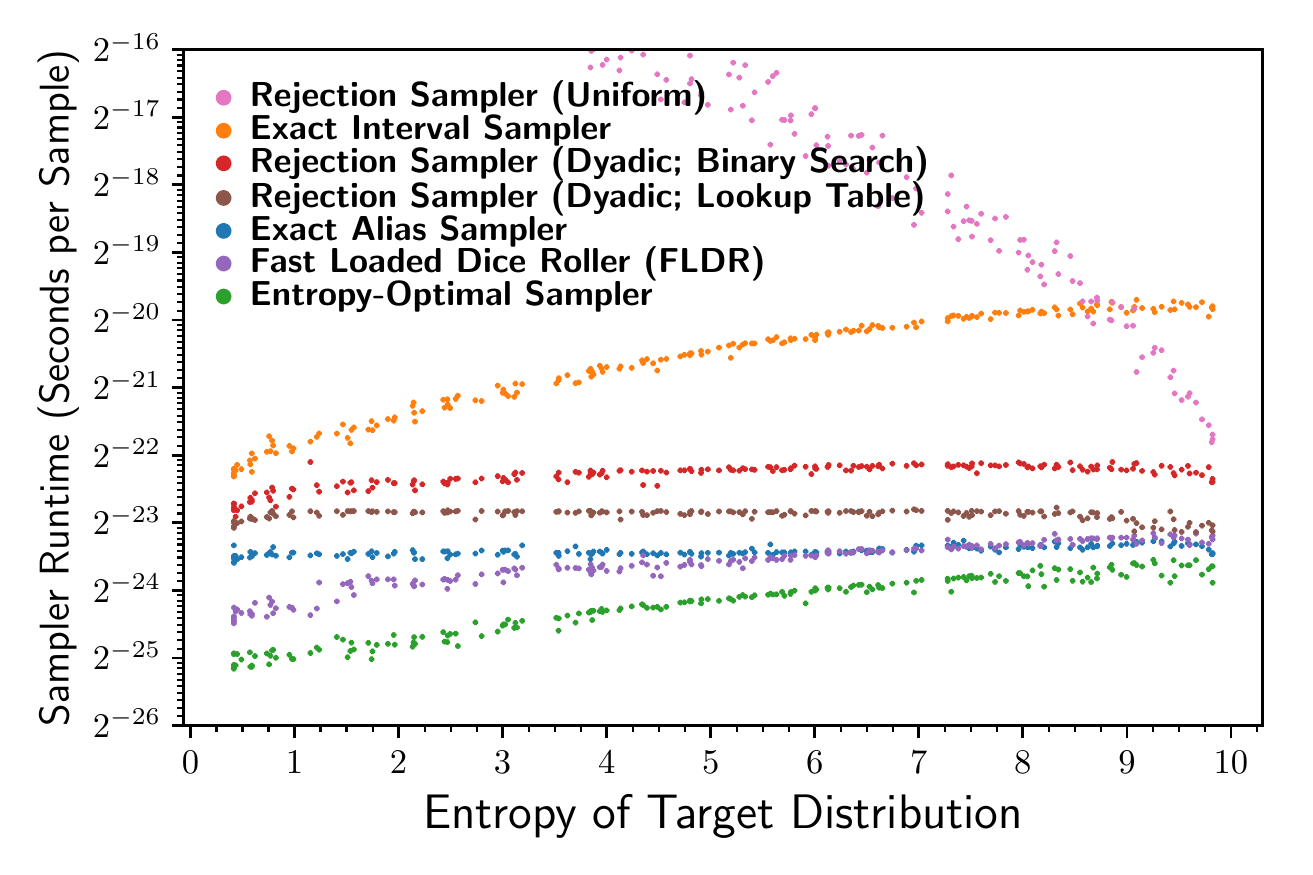}
\captionsetup{size=footnotesize, skip=0pt}
\caption{}
\label{fig:rate-entropy-100}
\end{subfigure}

\captionsetup{size=footnotesize, skip=5pt}
\caption{Comparison of memory and runtime performance for sampling 500
    random frequency distributions over $n\,{=}\,1000$ dimensions with
    sum $m\,{=}\,40000$, using FLDR and six baseline exact
    samplers.
    \subref{fig:memory-rate-100} shows a scatter plot of the sampler runtime
    (x-axis; seconds per sample) versus sampler memory (y-axis;
    bytes); and~\subref{fig:rate-entropy-100} shows how the sampler
    runtime varies with the entropy of the target distribution, for
    each method and each of the 500 distributions.}
\end{figure*}

\section{EMPIRICAL EVALUATION}
\label{sec:empirical-measurements}

We next empirically evaluate the memory, runtime, preprocessing, and
  entropy properties of the Fast Loaded Dice Roller from
  Section~\ref{sec:fast-loaded-dice-roller} and compare them to the
  following six baseline algorithms which, like FLDR, all produce
  exact samples from the target distribution and operate in the random
  bit model:
\begin{enumerate}[label=(\roman*),itemsep=0pt]
  \item entropy-optimal sampler~\citep{knuth1976},
  using a variant of Alg.~\ref{alg:fldr} (lines
  \ref{algline:loop-init}--\ref{algline:loop-end});
  \item rejection sampler with uniform proposal
  (Alg.~\ref{alg:rejection-uniform});
  \item rejection sampler with dyadic proposal~\citep{devroye1986},
  using a lookup table
  (Alg.~\ref{alg:rejection-dyadic-lookup-table});
  \item rejection sampler with dyadic proposal~\citep{devroye1986},
  using binary search
  (Alg.~\ref{alg:rejection-dyadic-binary-search});
  \item exact interval sampler~\citep{han1997}, using
  Alg.~1 of \citet{devroye2015};
  \item exact alias sampler~\citep{walker1977},
  using entropy-optimal uniform and Bernoulli sampling
  \citep{lumbroso2013} and the one-table implementation~\citep{vose1991}.
\end{enumerate}
All algorithms were implemented in C and compiled with \texttt{gcc}
  level 3 optimizations, using Ubuntu 16.04 on AMD Opteron 6376 1.4GHz
  processors.\footnote{All experiments in this section use target
  distributions with integer weights. We note that the
  reference implementations of FLDR in C and Python additionally
  contain preprocessing algorithms for exact sampling given \mbox{IEEE 754}
  floating-point weights. All samplers and experiments are at
  \url{https://github.com/probcomp/fast-loaded-dice-roller}.}

\subsection{Sampler Memory and Runtime}
\label{subec:empirical-measurements-exact}

We defined 100 frequency distributions $(a_1,\dots,a_n)$ over $n=100$
  dimensions which sum to $m=40000$, randomly chosen with entropies equally spaced
  from $0.034$ to $6.64 \approx \log{100}$ bits.
For each sampling algorithm and each distribution, we measured
  \begin{enumerate*}[label=(\roman*)]
  \item the size of the data structure created during
    preprocessing; and
  \item the wall-clock time taken to generate one million random
  samples.
  \end{enumerate*}
Fig.~\ref{fig:memory-rate-100} shows a scatter plot of the sampler
  memory (y-axis, in bytes) and sampler runtime (x-axis, in seconds
  per sample) for each algorithm and for each of the 100 distributions
  in the benchmark set, and Fig.~\ref{fig:rate-entropy-100} shows a
  scatter plot of the sampler runtime (y-axis, in seconds per
  sample) with the entropy of that target distribution (x-axis, in
  bits).

The runtime of FLDR (purple) most closely follows the runtime of the
  optimal sampler (green), while using up to $16000$x less
  memory---the memory improvement of FLDR grows at an exponential rate as $m$
  increases (Fig.~\ref{def:ddg-depth}).
In addition, for low-entropy distributions (bottom-left part of purple curve),
  FLDR uses even less memory than
  the linear bound from Thm.~\ref{thm:rejection-dyadic-ddg-depth}.

The lookup table rejection sampler (brown) uses up to $256$x more memory
  and is up to $4$x slower than FLDR, since it
  draws a constant $k=16$ bits/sample and uses a large size-$m$
  table---the memory improvement of FLDR again grows at an exponential rate as $m$
  increases.
The binary search rejection sampler (red) uses up to $32$x less
  than FLDR since it only stores running sums, but has up to $16$x slower
  runtime due to the cost of binary search---this
  runtime factor grows at a logarithmic rate as $n$ increases.
Rejection sampling with a uniform proposal (pink) performs poorly at
  low-entropy distributions (many rejections) and moderately at higher
  entropies where the target distribution is more uniform.

It is worthwhile to note that the \citeauthor{han1997} interval
  sampler (orange) has a tighter theoretical upper bound on entropy
  gap than FLDR ($3$ bits versus $6$ bits).
However, FLDR is up to 16x faster in our experiments,
  since we can directly simulate the underlying
  DDG tree using Alg.~\ref{alg:fldr}.
In contrast, implementations of the interval sampler in the
  literature for unbiased
  sources do not sample the underlying DDG tree, instead using
  expensive integer divisions and binary search in the main
  sampling loop~\citep{han1997,uyematsu2003,devroye2015}.
In addition, the array on which binary search is performed changes
  dynamically over the course of sampling.
To the best of our knowledge, unlike with FLDR, there is no
  existing implementation of interval sampling that directly
  simulates the underlying DDG tree so as to fully leverage its
  entropy efficiency.

The alias method (blue) is the most competitive baseline, which is up
  to 2x slower than FLDR (at low entropies) while using
  between 1x (at low-entropy distributions) and 8x less memory (at
  high entropies) to store the alias table.
While the alias method is commonly said to require constant runtime, this
  analysis only holds in the real RAM model and typical floating-point
  implementations of the alias method have non-zero sampling error.
For producing exact samples in the random bit model, the alias method
  requires (on average) between $\log{n}$ and $\log{n}+1$ bits to
  sample a uniform over $\set{1,\dots,n}$
  and two bits to sample a Bernoulli, which gives a
  total of $\log{n} + 3$ bits/sample, independently of $H(p)$
  (horizontal blue line in Fig.~\ref{fig:rate-entropy-100}).
In contrast, FLDR requires at most $H(p) + 6$ bits on average,
  which is less than alias sampling whenever
  $H(p) \ll \log{n}$.
For fixed $n$, the constant rate of the alias sampler corresponds to
  the ``worst-case'' runtime of FLDR: in
  Fig.~\ref{fig:rate-entropy-100}, the gap between purple (FLDR) and
  blue (alias) curves is largest at lower entropies and narrows as
  $H(p)$ increases.

%!TEX root=../paper.tex
\begin{figure}[t]
\includegraphics[width=\linewidth]{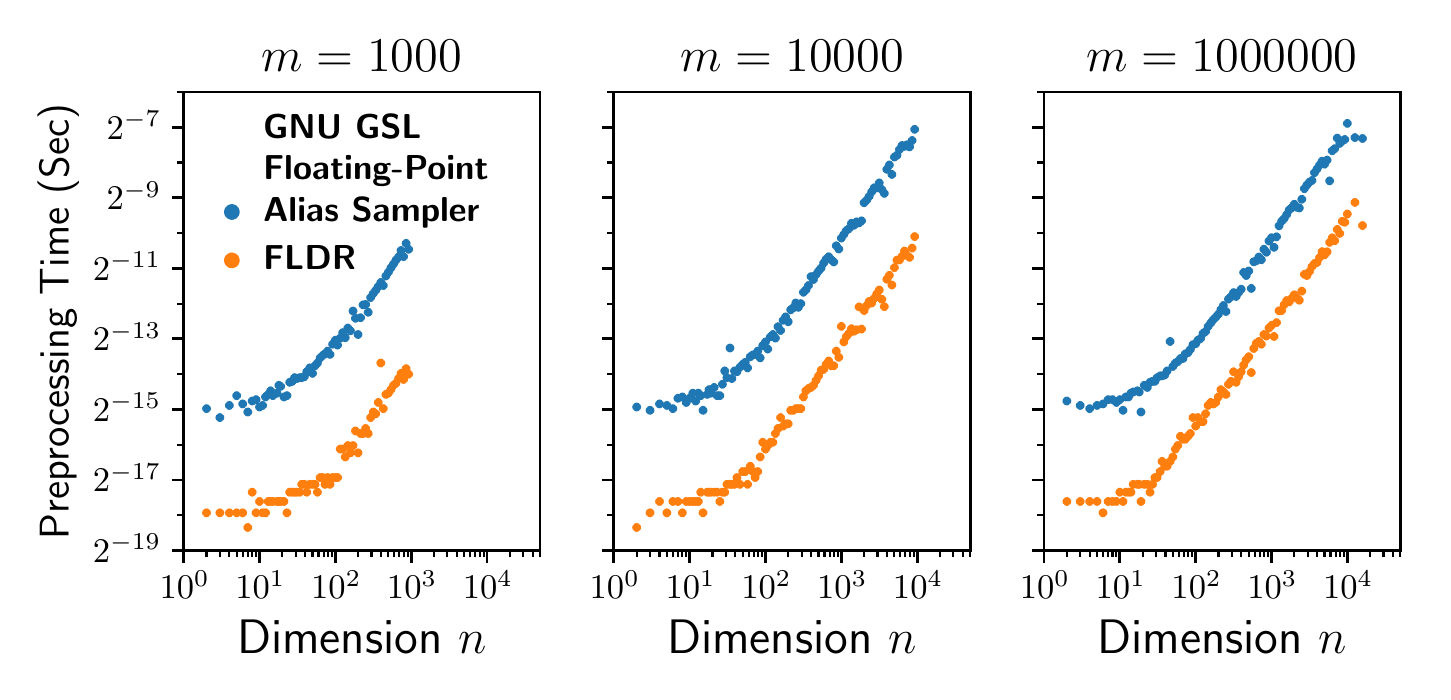}
\captionsetup{size=footnotesize, skip=0pt}
\caption{Comparison of the preprocessing times (y-axes; wall-clock seconds)
    of FLDR with those of the alias sampler, for distributions
    with dimension ranging from $n = 10^0, \dots 2\times10^4$ (x-axes) and
    normalizers $m=1000$, $10000$, and $1000000$ (left, center, and
    right panels, respectively).}
\label{fig:dimension-preprocess}
\end{figure}

\subsection{Preprocessing Time}
\label{subsec:empirical-measurements-preprocessing}

We next compared the preprocessing time of FLDR
  (Alg.~\ref{alg:fldr},
  lines~\ref{algline:fldr-preprocess-start}--\ref{algline:fldr-preprocess-end})
  for varying $(n,m)$ with that of the alias
  sampler~\citep{walker1977}, which is the most competitive baseline method.
To measure the preprocessing time of the alias method, we used the
  open-source implementation in the C GNU Scientific Library (GSL)%
  \footnote{The \texttt{gsl\_ran\_discrete\_preproc}
  function from the \texttt{gsl\_randist} GSL library implements
  the $O(n)$ alias table preprocessing algorithm from \citet{vose1991}.}.
Fig.~\ref{fig:dimension-preprocess} shows a log-log plot of the
  preprocessing time (y-axis; wall-clock seconds) and dimensions
  (x-axis; $n$) for distributions with
  $m=1000$, $10000$, $1000000$ (panels left to right).
Our C implementation of FLDR (orange) has a lower preprocessing
  time than the GSL alias sampler (blue) in all these regimes.
Since the matrix $H$ constructed during FLDR preprocessing
  has $n+1$ rows and $\log{m}$ columns, the gap between the
  two curves narrows (at a logarithmic rate) as $m$ increases.
On a 64-bit architecture we may assume that
  $m < 2^{64}$ (i.e., \texttt{unsigned long long} in C) and so
  the $n\log{m} \approx 64n$ preprocessing time of FLDR is highly
  scalable, growing linearly in $n$.

\subsection{Calls to Random Number Generator}
\label{sec:empirical-measurements-prng}

This paper has emphasized exact sampling in the random bit model,
  where the sampling algorithm lazily draws a random bit
  $B\,{\sim}\,\textsc{Flip}$ on demand.
As discussed in Section~\ref{sec:random-bit-model}, most sampling
  algorithms in existing software libraries operate under the real RAM model
  and approximate an ideal uniform variate $U\,{\sim}\,\uniform([0,1])$
  using a high-precision floating-point number.
Floating-point samplers produce non-exact samples---both as $U$ is not exactly
  uniform and as arithmetic operations involving $U$ (such as division) are
  non-exact.
Further, these implementations can be highly wasteful of computation.
(As an illustrative example, sampling a fair coin requires only one random
  bit, but comparing $U < 0.5$ in floating-point
  consumes a full machine word, e.g., 64 pseudo-random bits, to generate $U$.)
Following~\citet{lumbroso2013}, our implementation of $\flip$
  maintains a buffer of 64 pseudo-random bits.
Table~\ref{tab:prng-calls} shows a comparison of the number of calls
  to the pseudo-random number generator (PRNG) and wall-clock time for
  generating $10^6$ samples from 1000-dimensional
  distributions with various entropies,
  using FLDR and floating-point samplers (we have conservatively
  assumed that the latter makes exactly one PRNG call per sample).
The results in Table~\ref{tab:prng-calls} highlight that, by calling
  the PRNG nearly as many times as is information-theoretically
  optimal (Thm.~\ref{thm:rejection-dyadic-ddg-entropy}), FLDR
  spends significantly less time calling the PRNG than do floating-point
  samplers (with the added benefit of producing \mbox{exact samples}).

%!TEX root=../paper.tex

\begin{table}[t]
\footnotesize
\captionsetup{size=footnotesize, skip=0pt}
\caption{Number of PRNG calls and wall-clock time
  when drawing $10^6$ samples from $n=1000$ dimensional
  distributions, using FLDR \& approximate floating-point samplers.}
\label{tab:prng-calls}
\begin{tabular*}{\linewidth}{@{\extracolsep{\fill}}lrrr}
\toprule
\multirow{2}{*}{Method} & Entropy & Number of & PRNG Wall \\
~ & (bits) & PRNG Calls &  Time (ms) \\
\midrule
\multirow{4}{*}{\shortstack[l]{FLDR}}
  & 1 & 123,607 & 3.69 \\
% ~ & 2 & 165,335 & 3.88 \\
~ & 3 & 182,839 & 4.27 \\
% ~ & 4 & 234,443 & 6.08 \\
~ & 5 & 258,786 & 5.66 \\
% ~ & 6 & 288,997 & 8.60 \\
~ & 7 & 325,781 & 7.90 \\
% ~ & 8 & 352,015 & 9.10 \\
~ & 9 & 383,138 & 8.68 \\ \midrule
% ~ & 10  & 412,676 & 7.78 \\ \midrule
Floating Point &  all & 1,000,000 & 21.51 \\ \bottomrule
\end{tabular*}
\end{table}

\section{CONCLUSION}
\label{sec:conclusion}

This paper has presented the Fast Loaded Dice Roller, a
  new method for generating discrete random variates.
The sampler has near-optimal entropy consumption, uses a linear amount
  of storage, and requires linear setup time.
Due to its theoretical efficiency, ease-of-implementation
  using fast integer arithmetic,  guarantee of generating exact
  samples, and high performance in practice,
  we expect FLDR to be a valuable addition to the suite
  of existing sampling algorithms.

\bibliography{paper}

\clearpage
\appendix
%!TEX root=./paper.tex

% \input{figures/entropy-rate}

\section{PROOFS}
\label{appx:proofs}

This appendix contains the proofs of the theorems from
Section~\ref{sec:knuth-yao-complexity}, which are adapted from
\citet[Section~3]{saad2020popl} and included here for completeness.

\begin{proposition}[Proposition~\ref{prop:numsys-bases} in main text]
\label{prop:numsys-bases-appx}
For integers $k$ and $l$ with $0 \le l \le k$,
  define $Z_{kl} \defas 2^k - 2^{l}\Indicator_{l < k}$.
Then
  \begin{align*}
  \NumSys{kl} =
  \left\lbrace
    \frac{0}{Z_{kl}},
    \frac{1}{Z_{kl}},
    \dots,
    \frac{Z_{kl}-1}{Z_{kl}},
    \frac{Z_{kl}}{Z_{kl}}\Indicator_{l < k}
  \right\rbrace.
  \end{align*}
\end{proposition}

\begin{proof}
For $l=k$, the number system
  $\NumSys{kl} = \NumSys{kk}$
  is the set of dyadic rationals less than one
  with denominator $Z_{kk} = 2^k$.
For $0 \le l < k$, any $x \in \NumSys{kl}$ when written in
  base $2$ has a (possibly empty)
  non-repeating prefix and a non-empty infinitely repeating
  suffix, so that
  $x$ has binary expansion
  $(0.b_1\dots b_l\overline{s_{l+1}\dots s_k})_2$.
  Now,
  \begin{align*}
  2^l(0.b_1\dots b_l)_2
    = (b_1\dots b_l)_2
    = \textstyle\sum_{i=0}^{l-1}b_{l-i}2^{i}
  \end{align*}
  and
  \begin{align*}
  (2^{k-l}-1)(0.\overline{s_{l+1}\dots s_k})_2
    &= (s_{l+1}\dots s_k)_2 \\
    &= \textstyle\sum_{i=0}^{k-(l+1)}s_{k-i}2^{i}
  \end{align*}
  together imply that
  \begin{align*}
  x &= (0.b_1\dots b_l)_2 + 2^{-l}(0.\overline{s_{l+1}\dots s_k})_2 \\
    &=\frac
      {(2^{k-l}-1)\sum_{i=0}^{l-1}b_{l-i}2^{i} + \sum_{i=0}^{k-(l+1)}s_{k-i}2^{i}}
      {2^{k}-2^{l}}.
      \\[-\normalbaselineskip]\tag*{\qedhere}
  \end{align*}
\end{proof}

\begin{remark}
\label{remark:encoded-k-plus-one}
When $0 \le l \le k$, we have $\NumSys{kl} \subseteq \NumSys{k+1,l+1}$,
  since if $x \in \NumSys{kl}$
  then Proposition~\ref{prop:numsys-bases-appx}
  furnishes an integer $c$ such that
  $x = c/(2^k-2^l\Indicator_{l<k})
    = 2c/(2^{k+1}-2^{l+1}\Indicator_{l<k})
    \in \NumSys{k+1,l+1}$.
Further, for $k\ge 2$, we have
  $\NumSys{k,k-1} \setminus \set{1} = \NumSys{k-1,k-1} \subseteq \NumSys{kk}$,
  since any repeating suffix with exactly one
  digit can be folded into the prefix (except when the prefix
  and suffix are all ones).
\end{remark}

\begin{theorem}[Theorem~\ref{thm:k-bit-bases} in main text]
\label{thm:k-bit-bases-appx}
Let $T$ be an entropy-optimal DDG tree with a non-degenerate
  output distribution $(p_i)_{i=1}^{n}$
  for $n > 1$.
The depth of $T$ is the smallest integer $k$ such that there exists
  an integer $l \in \set{0,\dots,k}$ for which all the $p_i$
  are integer multiples of $1/Z_{kl}$ (hence in $\NumSys{kl}$).
\end{theorem}

\begin{proof}
Suppose that $T$ is an entropy-optimal DDG tree and
  let $k$ be its depth (note that $k\ge 1$, as $k=0$ implies $p$ is
  degenerate).
Assume $n=2$.
From Theorem~\ref{thm:ddg-knuth-yao}, for each $i=1,2$, the
  probability $p_i$ is a rational number where the number of digits in
  the shortest prefix and suffix of the binary expansion
  (which ends in $\bar{0}$ if dyadic) is at most $k$.
Therefore, we can express the probabilities $p_1, p_2$ in terms of
  their binary expansions as
  \begin{align*}
  p_1 &= (0.b_1\dots b_{l_1}\overline{s_{l_1+1}\dots s_{k}})_2, \\
  p_2 &= (0.w_1\dots w_{l_2}\overline{u_{l_2+1}\dots u_{k}})_2,
  \end{align*}
where $l_i$ and $k-l_i$ are the number of digits in the
  shortest prefix and suffix, respectively, of the binary expansions
  of each $p_i$.

If $l_1 = l_2$ then the conclusion follows from
  Proposition~\ref{prop:numsys-bases-appx}.
If $l_1 = k-1$ and $l_2 = k$ then the conclusion
  follows from Remark~\ref{remark:encoded-k-plus-one}
  and the fact that $p_1\ne 1$, $p_2 \ne 1$.
Now, from Proposition~\ref{prop:numsys-bases-appx}, it suffices to establish
  that $l_1 = l_2 \asdef l$, so that $p_1$ and $p_2$ are both integer
  multiples of $1/Z_{kl}$.
Suppose for a contradiction that $l_1 < l_2$ and $l_1 \ne k-1$.
Write
  $p_1 = a/c$
and
  $p_2 = b/d$
where each summand is in reduced form.
By Proposition~\ref{prop:numsys-bases-appx}, we have
  $c = 2^k - 2^{l_1}$
  and
  $d = 2^k - 2^{l_2}\Indicator_{l_2<k}$.
Then as
  $p_1 + p_2 = 1$
we have
  $ad + bc = cd$.
If $c \neq d$ then either $b$ has a positive factor in common with
  $d$ or $a$ with $c$, contradicting the summands being in reduced form.
But $c=d$ contradicts $l_1 < l_2$.
%
% The minimality of $k$ follows from the hypotheses $T$.

The case where $n > 2$ is a straightforward extension of this argument.
\end{proof}

\begin{theorem}[Theorem~\ref{thm:depth-perfect-sampling} in main text]
\label{thm:depth-perfect-sampling-appx}
Suppose $p$ is defined by $p_i = a_i/m$ $(i=1,\dots,n)$,
  where $\sum_{i=1}^{n}a_i=m$.
The depth of any entropy-optimal sampler for
  $p$ is at most $m-1$.
\end{theorem}

\begin{proof}
By Theorem~\ref{thm:k-bit-bases},
  it suffices to find integers $k \le m-1$ and $l \le k$ such that
  $Z_{kl}$ is a multiple of $m$, which in turn implies that any
  entropy-optimal sampler for $p$ has a maximum depth of $m-1$.
\begin{enumerate}[
  label={Case \arabic*:},
  wide,
  ]

\item $Z$ is odd.
\label{case:trustlessness}
  Consider $k = m-1$.
  We will show that $m$ divides $2^{m-1} - 2^{l}$ for some $l$ such
    $0 \le l \le m-2$.
  Let $\phi$ be Euler's totient function,
    which satisfies $1 \le \phi(m) \le m-1 = k$.
  Then $2^{\phi(m)} \equiv 1 \pmod{m}$ as $\mathrm{gcd}(m,2)=1$.
  Put $l = m - 1 - \phi(m)$ and conclude
    that $m$ divides $2^{m-1} - 2^{m-1-\phi(m)}$.

\item $m$ is even. Let $t\ge1$ be the maximal power of $2$ dividing $m$,
  and write $m = m'2^t$.
  Consider $k = m'-1+t$ and $l = j + t$
    where $j = (m'-1) - \phi(m')$.
  As in the previous case applied to $m'$, we have that
    $m'$ divides $2^{m'-1} - 2^{j}$, and so $m$ divides
  $2^k - 2^l$.
  We have $0 \le l \le k$ as $1 \le \phi(m) \le m-1$.
  Finally, $k = m'+t - 1 \le m'2^t - 1 = m - 1$ as $t < 2^t$.
  \qedhere
\end{enumerate}
\end{proof}

\begin{theorem}[Theorem~\ref{thm:depth-perfect-sampling-prime} in main text]
\label{thm:depth-perfect-sampling-prime-appx}
Let $p$ be as in Theorem~\ref{thm:depth-perfect-sampling-appx}.
If $m$ is prime and 2 is a primitive root modulo $m$,
  then the depth of an entropy-optimal DDG tree for
  $p$ is $m-1$.
\end{theorem}

\begin{proof}
Since $2$ is a primitive root modulo $m$,
  the smallest integer $a$ for
  which $2^{a} - 1 \equiv 0 \pmod{m}$
  is precisely $\phi(m) = m-1$.
We will show that for any $k' < m-1$ there is no exact
  entropy-optimal sampler that uses $k'$ bits of precision.
By Theorem~\ref{thm:depth-perfect-sampling-appx}, if there were such a
  sampler, then $Z_{k' l}$ must be a multiple of $m$ for some
  $l \le k'$.
If $l < k'$, then $Z_{k' l} = 2^{k'}- 2^l$. Hence
  $2^{k'} \equiv 2^l \pmod{m}$ and so $2^{k'-l} \equiv 1 \pmod{m}$
  as $m$ is odd.
  But $k' < m-1 = \phi(m)$,
  contradicting the assumption that $2$ is a primitive root modulo $m$.
If $l = k'$, then $Z_{k' l} = 2^{k'}$,
  which is not divisible by $m$ since we have assumed that $m$ is odd
  (as $2$ is not a primitive root modulo $2$).
\end{proof}

\end{document}